\documentclass[11pt,letterpaper]{article}

\usepackage{graphicx}
\usepackage{subfigure}
\usepackage{psfrag}
\usepackage{fullpage}
\usepackage{amsfonts}
\usepackage{amsmath}
\usepackage{amssymb}
\usepackage{url}
\usepackage{wrapfig}

\newcommand{\mygraphsubfigure}[4]
{\subfigure[#4]{\hspace{0.5cm}\includegraphics[#2]{#3}\label{#1}\hspace{0.5cm}}}

\newtheorem{theorem}{Theorem} 
\newtheorem{definition}{Definition} 
\newtheorem{lemma}{Lemma} 

\newenvironment{proof}[1][Proof]{\begin{trivlist}
\item[\hskip \labelsep {\bfseries #1:}]}{\hspace{\fill}$\square$\end{trivlist}}

{\makeatletter
 \gdef\xxxmark{%
   \expandafter\ifx\csname @mpargs\endcsname\relax 
     \expandafter\ifx\csname @captype\endcsname\relax 
       \marginpar{xxx}
     \else
       xxx 
     \fi
   \else
     xxx 
   \fi}
 \gdef\xxx{\@ifnextchar[\xxx@lab\xxx@nolab}
 \long\gdef\xxx@lab[#1]#2{{\bf [\xxxmark #2 ---{\sc #1}]}}
 \long\gdef\xxx@nolab#1{{\bf [\xxxmark #1]}}
 \long\gdef\xxx@lab[#1]#2{}\long\gdef\xxx@nolab#1{}%
 \gdef\turnoffxxx{\long\gdef\xxx@lab[##1]##2{}\long\gdef\xxx@nolab##1{}}%
}

\usepackage
  [breaklinks,bookmarks,bookmarksnumbered,bookmarksopen,bookmarksopenlevel=2]
  {hyperref}
{\makeatletter \hypersetup{pdftitle={\@title}}}

\let\realbfseries=\bfseries
\def\bfseries{\realbfseries\boldmath}

\def\captionfont{\sl\small}
\def\captionlabelfont{\bf\small}
{\makeatletter
 \global\let\plainfont@makecaption=\@makecaption
 \long\gdef\@makecaption#1#2{%
   \plainfont@makecaption{\captionlabelfont #1}{\captionfont #2}}}

\newif\ifabstract
\abstracttrue
\abstractfalse
\newif\iffull
\ifabstract \fullfalse \else \fulltrue \fi

\def\compactify{\itemsep=0pt \topsep=0pt \partopsep=0pt \parsep=0pt}
\let\latexusecounter=\usecounter
\newenvironment{itemize*}
  {\begin{itemize}\compactify}
  {\end{itemize}}
\newenvironment{enumerate*}
  {\def\usecounter{\compactify\latexusecounter}
   \begin{enumerate}}
  {\end{enumerate}\let\usecounter=\latexusecounter}

\let\epsilon=\varepsilon

\newcommand{\dominates}{path-expands }

\newcommand{\len}{\textrm{len}}
\newcommand{\ord}{\mathop{\rm Ord}\nolimits}
\newcommand{\annot}[2]{\mathop{\rm Annot}\nolimits_{#2}(#1)}

\newcommand{\dir}{\mathbf{dir}}
\newcommand{\defword}{\emph}
\newcommand{\RR}{\mathbb{R}}
\newcommand{\V}{V}
\newcommand{\nconf}{\mathop{\rm NConf}\nolimits}
\newcommand{\LL}{\mathcal{L}}
\newcommand{\conf}{\mathop{\rm Conf}\nolimits}
\newcommand{\asconf}[2]{\annot{\nconf_{#1}(#2)}{#2}}
\newcommand{\cgconf}[2]{C(#2)}
\newcommand{\igconf}[2]{A(#2)}
\newcommand{\rgconf}[2]{E(#2)}
\newcommand{\gconf}[2]{E(#2)}
\newcommand{\tmin}{\theta_{\mathrm{min}}}
\newcommand{\lmin}{l_{\mathrm{min}}}

\newcommand{\orel}{\succeq}

\begin{document}  
\psfragscanon




\title{A Generalized Carpenter's Rule Theorem \\ for Self-Touching Linkages}

\author{%
  Timothy G. Abbott%
    \thanks{MIT Computer Science and Artificial Intelligence Laboratory,
      32 Vassar St., Cambridge, MA 02139, USA,
      \protect\url{{tabbott,edemaine}@mit.edu}}
    \thanks{Partially supported by an NSF Graduate Research Fellowship
            and an MIT-Akamai Presidential Fellowship.}
\and
  Erik D. Demaine\footnotemark[1]
    \thanks{Partially supported by NSF CAREER award CCF-0347776,
            DOE grant DE-FG02-04ER25647, and AFOSR grant FA9550-07-1-0538.}
\and
  Blaise Gassend%
    \thanks{Exponent Failure Analysis Associates,
      149 Commonwealth Drive, Menlo Park, CA 94025, USA,
      \protect\url{blaise.gassend@m4x.org}.
      Work done while at MIT.}
}
\date{}

\maketitle

\begin{abstract}

The Carpenter's Rule Theorem states that any chain linkage in the
plane can be folded continuously between any two configurations while
preserving the bar lengths and without the bars crossing.  However,
this theorem applies only to strictly simple configurations, where
bars intersect only at their common endpoints.  We generalize the
theorem to self-touching configurations, where bars can touch but not
properly cross.  At the heart of our proof is a new definition of
self-touching configurations of planar linkages, based on an annotated
configuration space and limits of nontouching configurations.  We show
that this definition is equivalent to the previously proposed
definition of self-touching configurations, which is based on a
combinatorial description of overlapping features.  Using our new
definition, we prove the generalized Carpenter's Rule Theorem using a
topological argument.  We believe that our topological methodology
provides a powerful tool for manipulating many kinds of self-touching
objects, such as 3D hinged assemblies of polygons and rigid origami.
In particular, we show how to apply our methodology to extend to
self-touching configurations universal reconfigurability results for
open chains with slender polygonal adornments, and single-vertex rigid
origami with convex cones.

\end{abstract}

\clearpage

\section{Introduction}



In the mathematics of geometric folding
\cite{O'Rourke-1998,JCDCG2000b,OSME2001,FUCG_MSRI,FUCG}, a common
idealization is to model the underlying real-world object---a
mechanical linkage, robotic arm, protein, piece of paper, or another
object or surface---as having zero thickness.  The rods or bars that
make up a linkage become perfect mathematical line segments of fixed
length; the joints or hinges that connect them become mathematical
points; a piece of paper can be folded repeatedly ad infinitum.
While these idealizations are not entirely realistic (see
\cite{Gallivan-2002}), the zero-thickness model has led to a wealth of
powerful theorems that rarely abuse the lack of thickness and are
therefore practical.

\begin{wrapfigure}{r}{3.2in}
  \centering
  \vspace*{-0.5ex}
  \includegraphics[scale=0.9]{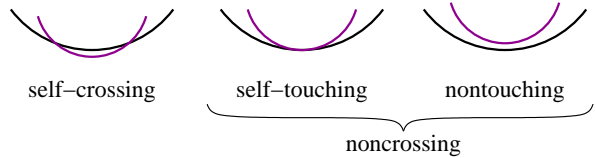}
  \vspace*{-1ex}
  \caption{The different types of configurations.}
  \vspace*{-0.5ex}
  \label{fig:vocabulary}
\end{wrapfigure}
Almost all forms of folding forbid folding objects from crossing,
matching a natural physical constraint, but at the same time allow
folding objects to touch.  Figure~\ref{fig:vocabulary} illustrates the
distinction between touching and crossing.  For example, overlapping
multiple layers of paper enables origamists to form arbitrarily
complicated shapes, both in practice and in theory~\cite{CGTA2000}.

Touching is easy to model for objects with positive thickness: allow
the boundaries, but not the interiors, to intersect.  But in the
zero-thickness model, formally distinguishing between touching and
crossing is difficult.  In particular, when two portions of the object
overlap, the geometry alone is insufficient to distinguish which
portion is on top of which.  The approach taken so far to resolving
the ambiguity is to express the information missed by the geometry
with additional combinatorial information.  A simple example is map
folding of an $m \times n$ grid of squares \cite{MapFolding}.  In this
context, the geometry of the squares is completely determined,
independent of the folding: in any successful folding that uses all
the creases, all of the squares will end up on top of each other, with
orientations specified by a checkerboard pattern in the grid.  The
folding itself can be specified by a purely combinatorial object: the
permutation of the panels that describes their total order in the
folding.  The challenge is to determine what constraints on this
combinatorial object correspond to the paper not self-crossing.  A
generalization of this approach is essentially the one taken by
\cite{PaperReachability_CCCG2004,FUCG} for defining general origami.

There are two concerns with this type of approach.

First, how do we know that the combinatorial definition corresponds to
the intended meaning of self-touching configurations?  The
combinatorial definitions inherently lack geometric intuition, so it
is hard to ``feel'' that they are correct, even though we believe they
are.

Second, how do we manipulate these definitions to prove interesting
theorems?  The complexity of the definitions makes them hard to use.
While some problems were successfully attacked in
\cite{InfinitesimallyLocked_LasVegas,PaperReachability_CCCG2004}, many
other problems about self-touching configurations remain open.  An
alternate, equivalent definition would give a new way to examine and
attack these problems.


Recently, touching has been studied for both linkages and origami.  In
the context of linkages, Connelly et
al.~\cite{InfinitesimallyLocked_LasVegas} show that self-touching
configurations of linkages could be used to prove theorems about
nearby non-self-touching perturbations.  This result essentially
reduces proving a planar linkage to be locked to an automatic,
algorithmic procedure, whereas previous arguments that dealt solely
with non-self-touching configurations were tedious and ad-hoc.  To do
this, \cite{InfinitesimallyLocked_LasVegas} introduced a combinatorial
definition of touching linkages that we will describe later.

Many results in computational origami construct folded states with the
desired properties, but do not show that the state can be reached by a
continuous folding motion.  Demaine et
al.~\cite{PaperReachability_CCCG2004} showed that this was always
possible.  To do this, they defined origami using a combinatorial
definition to handle self-touching folded states and their folding
motions.  This combinatorial definition turned out to be tedious to
work with.

\paragraph{Our results.}
In this paper, we study the self-touching analog of the Carpenter's
Rule Theorem, posed at FOCS 2000 \cite{Linkage}.  Consider a polygon
or open polygonal chain in the plane, where the edges represent rigid
bars of fixed length and the vertices represent hinges that can take
arbitrary angles.  Connelly et al.~\cite{Linkage} proved the
Carpenter's Rule Theorem: every such linkage can be unfolded to a
convex configuration while preserving connectivity, edge lengths, and
without self-crossing; see also
\cite{Streinu-2005,Cantarella-Demaine-Iben-O'Brien-2004} for more algorithmic
approaches.  But it remained open whether this result held when the
original configuration was self-touching (but not self-crossing).  We
solve this open problem, proving that the Carpenter's Rule Theorem
extends to self-touching polygons and polygonal chains: every such
linkage can be convexified starting from any (possibly self-touching)
configuration.

The basis for this result is a new technique for defining self-touching
configurations.  Our approach is based on the intuition that
self-touching configurations are limits of non-self-touching configurations.
This intuition may seem obvious, but in its literal form, it is false.
In the example of map folding, taking the limit to zero separation
still, in the end, discards all of the information about the folding.
Nonetheless, when working with self-touching configurations, people draw
configurations with overlapping layers separated slightly for visibility,
and imagine the limit as those separations go to zero.

We show how to turn this intuitive idea into a definition.  Our main
idea is to \emph{annotate} the geometry of the configuration with
additional continuous information.  Previous combinatorial definitions
could add annotations only at places where self-touching occurs, as
needed to resolve ambiguity.  In contrast, the topology of our
configuration space places self-touching configurations near
nontouching configurations.  We therefore annotate all configurations,
independent of whether they are self-touching.  Generally, the annotations are
made up of the output of an order function applied to each ordered
pair of independently mobile parts of the object to be modeled.  For a
linkage, bars are the independently mobile parts; for paper, each
point of the paper is an independently mobile part.

\begin{wrapfigure}{r}{1.8in}
  \centering
  \mygraphsubfigure{fig:need-stretch-a}{scale=0.5}{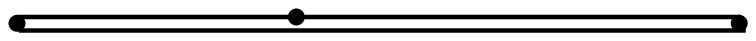}{Self-touching
  configuration.}
  \mygraphsubfigure{fig:need-stretch-b}{scale=0.5}{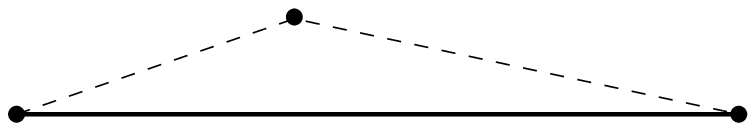}{Stretched
  configuration.}
  \caption{This self-touching configuration of a degenerate triangle
    linkage needs some added flexibility to be a limit of nontouching
    configurations.}
  \label{fig:need-stretch}
\end{wrapfigure}

Taking limits of annotated configurations is unfortunately
insufficient to get all self-touching configurations.
Figure~\ref{fig:need-stretch-a} shows a linkage that has no
nontouching configurations.  In order to get these configurations, we
allow flexibility in the limit-taking.  That is, we allow the object
to ``stretch'' while the limit is being taken.  In the case of
linkages, stretching means varying the length of the bars
(Figure~\ref{fig:need-stretch-b}).

We believe that uniform annotation and stretchiness can be used to
define self-touching configurations for a wide range of foldable
objects.  In this paper, we use our limit-based definition to prove
the self-touching Carpenter's Rule Theorem using a primarily
topological argument.  At the end of the paper, we discuss how our
topological approach might be applied to origami, rigid origami, and
3D hinged assemblies of planar panels.

\paragraph{Outline.}
In Section~\ref{sec:linkageequivalence}, we review 2D linkages and
define $\epsilon$-related configurations of linkages.  Then, in
Section~\ref{sec:general} we present an order function designed for
linkages and apply it to self-touching linkage configurations. To
increase the credibility of this definition,
Section~\ref{sec:equivalence} proves it equivalent to the previously existing
combinatorial definition as well as a variant of our definition that allows
vertices to be split into a pair of vertices connected by a zero-length edge.
Section~\ref{sec:carpenter} then uses the new definition to prove the
self-touching Carpenter's Rule Theorem.
We generalize this universal reconfigurability result
to strictly slender polygonal adornments in Section~\ref{sec:slender}.
Finally, in Section~\ref{sec:extension} we discuss extensions of our
definition methodology to objects that cannot be modeled as 2D linkages.


\section{Linkage Preliminaries}
\label{sec:linkageequivalence}

This section introduces basic definitions that are important
throughout this paper.

\begin{definition}
A \defword{linkage} is a pair $\LL = (G, \ell)$ consisting of a graph
$G$ and a function $\ell : E(G) \to \RR_{\ge 0}$ assigning nonnegative
lengths to the edges.  We refer to the edges of $G$ as \emph{bars}.
\end{definition}

\begin{definition}
A \defword{configuration} $C$ of a linkage $\LL = (G, \ell)$ in the
plane is a map $C : \V(\LL) \to \RR^2$ obeying the length constraints,
so if $(v, w) \in E(G)$ then $|C(v) - C(w)| = \ell(v, w)$.  The set of
all such configurations is called the \emph{configuration space}
$\conf(\LL)$ of $\LL$.
\end{definition}

\begin{definition}
A \defword{nontouching configuration} of $\LL=(G, \ell)$ is a
configuration in which no two edges intersect except at endpoints, and
two endpoints coincide if and only if they are connected by a path in
$G$ of zero-length bars.  Let $\nconf(\LL) \subset \conf(\LL)$ be the
subspace of nontouching configurations.
\end{definition}

Simple linkages include \defword{open chains} and \defword{closed
chains} for which the underlying graph is a single path, or a single
loop, respectively.

\begin{wrapfigure}{r}{1.7in}
  \vspace*{-3ex}
  \includegraphics{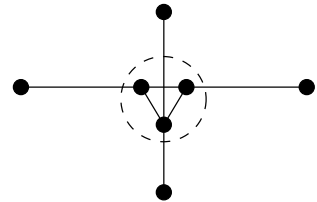}
  \vspace*{-1ex}
  \caption{This configuration, where the bars within the dotted circle
    have zero length, is nontouching, despite the apparent
    intersection between the upward-pointing edge and the topmost
    zero-length edge.}
  \label{fig:zero-length}
  \vspace*{-9ex}
\end{wrapfigure}

Our definition of linkage is unusual in that it allows zero-length
bars.  Allowing zero-length bars will be necessary for some of our
topological arguments, because without them, certain configuration
spaces are not closed.

The definition we have taken for nontouching may yield surprising
results with linkages having zero-length bars, as in
Figure~\ref{fig:zero-length}.  Our definition considers pairs of
vertices that are connected by a zero-length edge to be nontouching,
effectively merging the two endpoints of such edges into a single
vertex.

\subsection{$\bf{\epsilon}$-related Configurations}

\begin{definition}
For $\epsilon \ge 0$, we say two linkages $(G_1, \ell_1) $ and $(G_2,
\ell_2)$ are \defword{$\epsilon$-related} if $G_1=G_2$
and $|\ell_1(e)-\ell_2(e)| \le \epsilon$ for all $e \in E(G_1)$.
\end{definition}

\begin{definition}
For $\epsilon \ge 0$, an \defword{$\epsilon$-related configuration} of
a linkage $\LL$ is a configuration of a linkage $\LL'$ that is
$\epsilon$-related to $\LL$.
\end{definition}

The set of $\epsilon$-related configurations of $\LL$ is denoted by
$\conf_{\epsilon}(\LL)$.  In particular, $\conf_0(\LL)=\conf(\LL)$.
Similarly, $\nconf_{\epsilon}(\LL)$ is the set of nontouching
configurations of linkages $\epsilon$-related to $\LL$.

\subsection{Real Algebraic Geometry}
We use some results from real algebraic geometry,
In particular, all of the objects discussed in this work will
be semi-algebraic, and we use several topological properties of such sets.
For a comprehensive reference on the theory of semi-algebraic sets,
see \cite{Bochnak-Coste-Roy-1998}.
\xxx{also cite \cite{Munkres-2000}?}
We recall the definition:

\begin{definition}
A (real) \defword{semi-algebraic set} is subset of $\RR^n$ that is a
finite Boolean combination of sets of the form $\{x: f(x) * 0\}$ where
$f$ is a polynomial of $n$ variables, and $*$ is $=$ or $<$.
A function $f : \RR^m \to \RR^k$ is \defword{semi-algebraic} if its
graph $\{(x, f(x)) : x \in \RR^m\} \subseteq \RR^{m+k}$ is a
semi-algebraic set.
\end{definition}

We will use a number of important topological properties of
semi-algebraic sets.  For a comprehensive reference on the theory of
semi-algebraic sets, see \cite{Bochnak-Coste-Roy-1998}.
\xxx{also cite \cite{Munkres-2000}?}



\section{Noncrossing Configurations}
\label{sec:general}

\subsection{The Order Function}
In defining noncrossing configurations, we must allow two bars to
``overlap''.  When this happens, the geometry of the configuration
space does not have the information necessary to determine when they
cross.

We define an order function, $\ord$, that determines the relative
positions of two edges in noncrossing configurations.  Its key
property is that $\ord$ is continuous wherever the edges do not touch,
but $\ord$ has a discontinuity where two edges share a common segment.
The two different limits as the edges approach each other will encode
the relative position information.

\bigskip

\noindent
\begin{minipage}{3.7in}
\begin{definition}\parindent=1.5em

Let $e_1$ and $e_2$ be two oriented edges in the plane.  Using
coordinates where the $x$-axis is directed along $e_1$, define
$$d_{+}(e_1, e_2) = \len \{x \in e_1 \mid \exists y: y \ge 0, (x, y) \in
e_2\};$$ $$d_{-}(e_1, e_2) = \len \{x \in e_1 \mid \exists y: y \le 0,
(x, y) \in e_2\}.$$ $d_+$ can be thought of as the length of the
projection of the part of $e_2$ above $e_1$ onto the \textbf{segment}
$e_1$ (and similarly for $d_-$).  See Figure~\ref{fig:ord}.  Define
the \defword{order function} $\ord(e_1, e_2) = d_+(e_1,e_2)-d_-(e_1,
e_2)$.

\end{definition}
\end{minipage}\hfil
\begin{minipage}{2.4in}
  \makeatletter\def\@captype{figure}%
  \centering
  \includegraphics[scale=1.0]{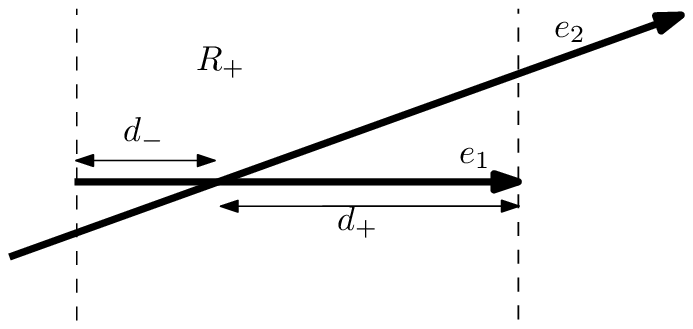}
  \caption{Defining the order function.\newline }
  \label{fig:ord}
\end{minipage}

The function $\ord$ is defined everywhere, but indeed is not continuous
when the two edges are tangent.  In particular, $\ord(e_1, e_1) = 0$,
while if $e_2$ is a slight parallel offset of $e_1$ by~$\pm \delta$,
$\ord(e_1,e_2)$ is $\pm \len(e_1)$, which is very different from~$0$.

\begin{lemma}
\label{lemma:ord-func-prop}

The order function has the following properties:

\begin{enumerate*}

\item $\ord$ is a semi-algebraic function.

\item $\ord$ is continuous over the subdomain of pairs of edges that
do not intersect in their interior.

\item Consider two sequences of oriented edges $e^n_1$ and $e^n_2$
converging to edges $e_1$ and $e_2$ respectively, such that $e^n_1$
and $e^n_2$ do not intersect in their interior for any $n$ (note that
$e_1$ and $e_2$ might intersect in their interior).  Then the sequence
$\ord(e^n_1, e^n_2)$ either converges to $\ord(e_1, e_2)$, or has at
most two accumulation points: $d$ and $-d$, where $d$ is the length of
the overlap between $e_1$ and $e_2$. \label{lembullet:limit-ord}

\end{enumerate*}

\end{lemma}

\ifabstract
\else
\begin{proof}
We prove that $d_+$ is a semi-algebraic function that is
continuous over edges that do not intersect in their interior;
then $d_-$ has these properties by a similar argument.  We use the
reference frame centered at the first vertex of $e_1$, with the $x$
axis directed along $e_1$.  In this reference frame, edge $e_1$
extends from $(0,0)$ to $(l,0)$.

Now, replace $e_2$ with its intersection with the region $R_+ =
\{(x,y) \mid y \ge 0, l \ge x \ge 0\}$.  If the intersection is empty
(a semi-algebraic condition), then $d_+=0$.  Otherwise, the
coordinates $(x_1, y_1)$ and $(x_2, y_2)$ are a semi-algebraic
function of the endpoints of $e_2 \cap R_+$, because whether $e_2$
intersects each boundary of $R_+$ can be determined using line-segment
intersections and thus can be tested with a boolean combination of
polynomial inequalities.  For each possible set of intersections, the
points of intersection can be semi-algebraically computed from the
original coordinates.

Clearly, $d_+=0$ if $l=0$.  The reader can check that
$d_+=\left|f(x_2/l)-f(x_1/l)\right|l$, where $f(z)=\max(\min(z,1),0)$.
$f$ is a boolean combination of polynomial inequalities, and thus
$d_+$ is a semi-algebraic function everywhere.

Because $d_+$ is uniformly zero when $e_2$ does not intersect $R_+$,
$d_+$ is continuous when $e_2 \cap R_+$ is empty.  Notice that our
formula for $d_+$ in terms of $f$ would still be true if we had
replaced $e_2$ with its intersection with the closed upper half plane
$H$ instead of $R_+$.  Because $f$ is continuous, $d_+$ is continuous
everywhere where $e_2 \cap H$ is a continuous function of the
coordinates of $e_2$.  Similarly, $d_+$ is continuous everywhere where
$e_2 \cap R_+$ is a continuous function of the coordinates of $e_2$.
Thus, $d_+$ is continuous except when both endpoints of $e_2 \cap R_+$
lie along the boundary of $H$ intersect the boundary of $R_+$.  This
exceptional case occurs only if $e_2$ and $e_1$ intersect in the
interior.  Thus $d_+$ is continuous everywhere except when $e_2$
intersects $e_1$ in the interior.

Now, consider sequences $e_1^n$ and $e_2^n$ defined in the statement
of Part~\ref{lembullet:limit-ord} of this lemma.  By continuity,
$\ord(e_1^n, e_2^n)$ converges to $\ord(e_1, e_2)$, unless the limits
$e_1$ and $e_2$ share a common interval.  If they do, notice that
$e_1^n$ and $e_2^n$ are nontouching, and so for sufficiently large
$n$, $\ord(e_1^n, e_2^n)$ must be within $\epsilon$ of either $l$ or
$-l$, where $l$ is the length of that common interval.  The result
follows.\end{proof}
\fi

\subsection{Annotations}

An annotated configuration is a configuration augmented with the
values of $\ord$ on each pair of edges in the configuration.

\begin{definition}
Let $C$ be a configuration or $\epsilon$-related configuration of a
linkage $\LL$ with graph $G=(V, E)$.  For $e_i = (u,v) \in E$, write
$C(e_i) = (C(u), C(v))$.

Let $Annot_{\LL}: \RR^{|V|} \rightarrow \RR^{|V|} \times
\RR^{|E|\times|E|}$ be defined by $\annot{C}{\LL} = (C, A)$, where
$A_{i,j} = \ord(C(e_i),C(e_j))$.\footnote{We assume that the edges
have been assigned some canonical orientation.}  In this case we call
$A$ the \emph{annotation} of $C$ and the pair $(C, A)$ an
\emph{annotated configuration}.
\end{definition}


\begin{definition}
The set of \defword{annotated nontouching configurations} is
$\annot{\nconf(\LL)}{\LL}$.  For $\epsilon \ge 0$, the set of
\defword{annotated nontouching $\epsilon$-related configurations} is
$\annot{\nconf_\epsilon(\LL)}{\LL}$.
\end{definition}

\begin{lemma}
\label{lemma:annot-func-prop}
The annotation function $Annot_{\LL}$ is injective, continuous over
nontouching configurations, and semi-algebraic.
\end{lemma}

\begin{proof}
The annotation function is injective, because the first component of
$\annot{C}{\LL}$ is $C$.  Because the annotation function simply applies
$\ord$ to all pairs of edges in $G$, and by definition, in nontouching
configurations no two edges intersect in their interior, the remaining
properties follow directly from Lemma~\ref{lemma:ord-func-prop}.
\end{proof}

\subsection{Noncrossing Configurations}
We are now ready to define noncrossing configurations in terms of
limits of nontouching configurations.

\begin{definition} \label{def:igconf}
A \defword{noncrossing configuration} of $\LL$ is an element of
$\conf_0(\LL) \times \RR^{|E(\LL)|\times|E(\LL)|}$ that is the limit
of a sequence of annotated nontouching configurations of linkages
$1$-related to $\LL$.  The space of noncrossing configurations of
$\LL$ is denoted $\igconf{0}{\LL}$.

Equivalently, $\igconf{0}{\LL} = (\conf_0(\LL) \times \RR^{|E(\LL)
  \times E(\LL)|}) \cap \overline{\asconf{1}{\LL}}$, where
  $\overline{X}$ denotes the topological closure of $X$.

\end{definition}

The following characterization implies that replacing $1$ with any
$\epsilon > 0$ in Definition \ref{def:igconf} would define the same
set.


\begin{lemma}\label{lem:igconf-property}$\igconf{0}{\LL} = \cap_{n=1}^\infty \overline{\annot{\nconf_{1/n}(\LL)}{\LL}}$.
\end{lemma}
\begin{proof}
Clearly, $\cap_{n=1}^\infty \overline{\annot{\nconf_{1/n}(\LL)}{\LL}}
\subset \overline{\asconf{1}{\LL}}$.  To see that $\cap_{n=1}^\infty
\overline{\annot{\nconf_{1/n}(\LL)}{\LL}} \subset \conf_0(\LL) \times
\RR^{|E(\LL)|\times|E(\LL)|}$, notice that the length of any bar $e$
in the left-hand side must differ from $\ell(e)$, the length of that
bar in $\LL$, by at most $1/n$ for all $n$, and thus must equal
$\ell(e)$, so that we in fact have an annotated configuration of
$\LL$.

Conversely, if $x \in \conf_0(\LL)$ is the limit of the sequence
$\alpha_1, \alpha_2, \ldots$, $\alpha_i \in
\annot{\nconf_1(\LL)}{\LL}$, then for any $n$, there exists $N_0(n)$
such that for any $k > N_0(n)$, $\alpha_k$ is an annotated
$1/n$-related configuration of $\LL$, i.e. $\alpha_k \in
\overline{\annot{\nconf_{1/n}(\LL)}{\LL}}$.  Because that set is closed,
the limit $x$ must also be in
$\overline{\annot{\nconf_{1/n}(\LL)}{\LL}}$.  Thus $\igconf{0}{\LL}
\subset \cap_{n=1}^\infty \overline{\annot{\nconf_{1/n}(\LL)}{\LL}}$,
as desired.
\end{proof}

\subsection{Semi-Algebraic}

\begin{theorem}
\label{th:conf-semialgebraic}

For $\epsilon \ge 0$ the following sets are semi-algebraic:
$\conf_{\epsilon}(\LL)$, $\nconf_{\epsilon}(\LL)$,
$\asconf{\epsilon}{\LL}$, $\igconf{\epsilon}{\LL}$.
\end{theorem}

\ifabstract
\else
\begin{proof}
We start with $\conf_{\epsilon}(\LL)$, which is defined by requiring
each bar length to be within $\epsilon$ of its length in $\LL$.  For a
bar between points $(x_i, y_i)$ and $(x_j, y_j)$, with a length $l_k$
in $\LL$, we have the following constraints:
\begin{eqnarray*}
(x_i-x_j)^2 + (y_i-y_j)^2 \le & (l_k+\epsilon)^2, &\\ (x_i-x_j)^2 +
(y_i-y_j)^2 \ge & (l_k-\epsilon)^2 & \textrm{ if $l_k \ge \epsilon$.}
\end{eqnarray*}
\noindent Because these conditions are semi-algebraic,
$\conf_{\epsilon}(\LL)$ is semi-algebraic.


$\nconf_{\epsilon}(\LL)$ is defined like $\conf_{\epsilon}(\LL)$, but
with additional nontouching constraints.  We re-use a strategy
presented in Equation~(3.6) of~\cite{InfinitesimallyLocked_LasVegas},
based on the following idea: if two bars do not intersect, then one of
the bars lies completely on one side of the other bar, i.e., both ends
of the first bar are on the same side of the other bar. The condition
in~\cite{InfinitesimallyLocked_LasVegas} has to be slightly changed by
making the inequalities in it strict, to prevent self-touching in
addition to self-intersection.

Unfortunately, this condition is too strong, as it prevents bars from touching at
their endpoints when the distance in the graph between the endpoints is
zero. If there is a path $v_{i_0}, \dots, v_{i_n}$ in the graph between two
vertices $v_{i_0}$ and $v_{i_n}$, then we can test that the distance
between them in the graph is zero using the equation
$$ \sum_{j=0}^{n-1} (x_{i_j}-x_{i_{j+1}})^2 + (y_{i_j}-y_{i_{j+1}})^2 = 0.$$
When this condition holds, we augment the strict nontouching condition by
explicitly allowing $v_{i_0}$ and $v_{i_n}$ to touch. We allow this as
long as the other vertex of each bar does not touch the other bar, or one
of the bars has zero length. All these conditions can be expressed using
polynomials.

Combining all these polynomial conditions, we find that
$\nconf_{\epsilon}(\LL)$ is semi-algebraic.  It follows that
$\asconf{\epsilon}{\LL}$ is semi-algebraic, because it is the image of
the semi-algebraic set $\nconf_{\epsilon}(\LL)$ under the
semi-algebraic annotation map (by Lemma \ref{lemma:annot-func-prop}).
Similarly, $\annot{\conf_{\epsilon}(\LL)}{\LL}$ is semi-algebraic.


Finally, $\igconf{\epsilon}{\LL}$ is the intersection of the
topological closure of the semi-algebraic set $\asconf{1}{\LL}$ with
the semi-algebraic set $\annot{\conf_{0}(\LL)}{\LL}$, and is thus
semi-algebraic.
\end{proof}
\fi

\section{Equivalent Definitions}\label{sec:equivalence}

In this section, we describe two definitions of noncrossing
configurations that are equivalent to our definition.  Section
\ref{ssec:reduced} describes a variation on our definition which is
useful for analyzing the linkages with very short edges.

Section \ref{sec:comb-equivalence} describes a translation of the
combinatorial definition of~\cite{InfinitesimallyLocked_LasVegas} into
our terminology, and proves that all three definitions are equivalent.

\subsection{Extended Linkages}\label{ssec:reduced}

In this section, we introduce \emph{extended linkages} and
\emph{extended noncrossing configurations}, and use them to define a
natural and seemingly larger class of noncrossing configurations.  We
use extended linkages as a tool in our proof that our noncrossing
configurations are the same as the noncrossing configurations
of~\cite{InfinitesimallyLocked_LasVegas}, which we will refer to as
combinatorial noncrossing configurations.

When defining noncrossing configurations as a limit of nontouching
configurations, we could allow vertices to be split into two vertices
with an edge of length at most $\epsilon$ between them.  In the limit
as $\epsilon \rightarrow 0$, the extra edges have length 0, and we
remerge their endpoints, so that the resulting configuration is
naturally an element of $\annot{\conf_{\epsilon}(\LL)}{\LL}$.  Linkage
extensions and reductions formalize this notion of splitting vertices
into two vertices separated by a zero-length edge.  A simple example
is shown in Figure~\ref{fig:extended-linkage-example}.

\begin{wrapfigure}{r}{1.8in}
  \centering
  \scalebox{0.5}{\includegraphics{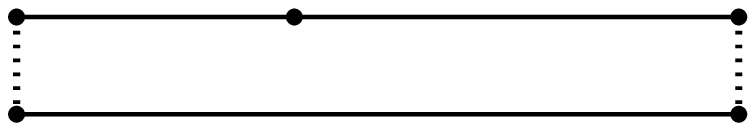}}
  \caption{Extended Linkage version of Figure \ref{fig:need-stretch}}
  \label{fig:extended-linkage-example}
  \vspace*{-2ex}
\end{wrapfigure}

This formulation is useful for analyzing linkages that have very short
edges, by understanding their self-touching limit.  This technique is
can be used to simplify a proof that there exists a locked orthogonal
tree~\cite{Charlton-Demaine-Demaine-Price-Tu-2008}.  In particular, the
orthogonal tree has horizontal edges of length at most $\epsilon$, but it is
easiest to argue it is rigid in the limiting case $\epsilon=0$
(where it has some zero-length bars), and then conclude that it is locked
for some small $\epsilon$ using the equivalence between extended
configurations and combinatorial configurations along with results
of~\cite{InfinitesimallyLocked_LasVegas}.


\begin{definition}
Given a linkage $\LL$ with an edge $e=(u,v)$ of length 0, we can
construct a \defword{single-step reduced linkage} in which $e$ has
been removed and $u$ and $v$ have been merged into a single vertex
(this may create multiple copies of some edges, but this is not a
functional change because those edges were already in the same location
before the two vertices were merged). A \defword{reduction} of a
linkage $\LL$ is the result of zero or more single-step reductions
starting from $\LL$.  An \defword{extension} of $\LL$ is a linkage
$\LL'$ of which $\LL$ is a reduction.
\end{definition}

One useful way to extend a linkage is to replace each vertex with one
vertex per incident edge, all connected together by 0-length bars.

Given an annotated configuration $C$ of a linkage $\LL$, a
configuration of the extended linkage $\LL'$ can be generated by
placing the fragments of a newly split vertex $v$ where $v$ was
before, and setting the annotations of the new zero-length edge with
other edges to 0.  In every configuration of $\LL'$, the two fragments
must be in the same location, so this defines a correspondence between
$\annot{\conf(\LL)}{\LL}$ and $\annot{\conf(\LL')}{\LL'}$.  However,
the $\epsilon$-related configurations of $\LL'$ are a larger class
than those of $\LL$.



\begin{definition}
An \emph{extended noncrossing configuration} of $\LL$ is the reduction
to $\LL$ of a noncrossing configuration of some extension $\LL'$ of
$\LL$.  The set of extended noncrossing configurations is denoted
$\gconf{0}{\LL}$.
\end{definition}

\subsection{Combinatorial Characterization of a Noncrossing Configuration}\label{sec:comb-equivalence}

\begin{wrapfigure}{r}{2.1in}
  \vspace{-3ex}
  \mygraphsubfigure{fig:combinatorial-def-1}{scale=0.5}{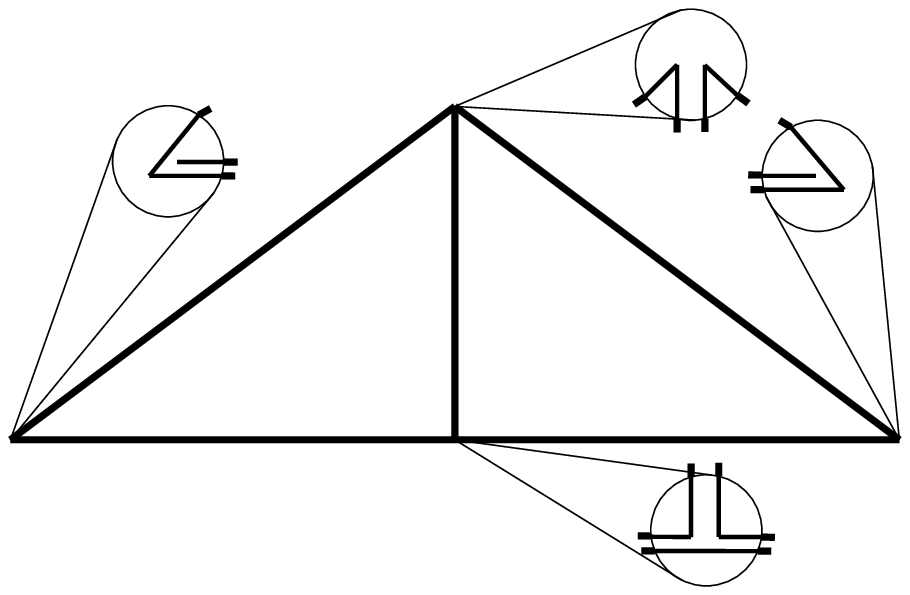}{Combinatorial
  noncrossing configuration.}
  \mygraphsubfigure{}{scale=0.5}{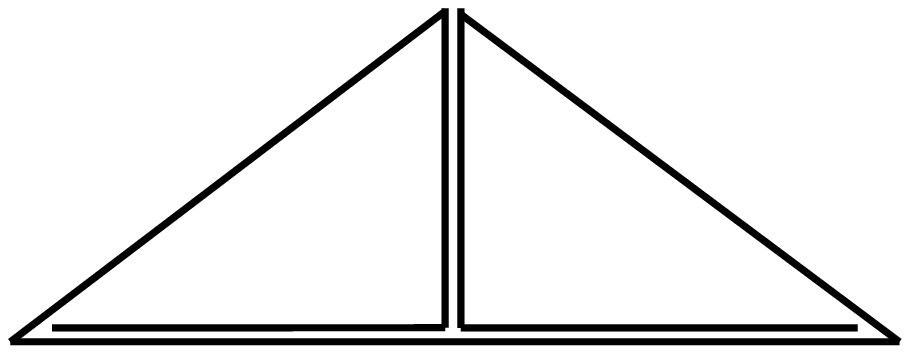}{Nearby $\epsilon$-related
  nontouching configuration.  A noncrossing configuration is a limit
  of these as $\epsilon \rightarrow 0$.} \centering
  \caption{Combinatorial noncrossing configuration example.}
  \label{fig:combinatorial-def}
  \vspace*{-2ex}
\end{wrapfigure}

So far we have defined a noncrossing configuration of a linkage as an
annotated configuration which is a limit of annotated nontouching
configurations.  This topological definition is markedly different
from combinatorial definition given
in~\cite{InfinitesimallyLocked_LasVegas}.  We shall now see that the
annotated configurations which are nontouching can be characterized
combinatorially as well as topologically.  The combinatorial
characterization will also allow us to see that in fact,
$\gconf{0}{\LL}=\igconf{0}{\LL}$, so extensions are not necessary to
generate all nontouching linkage configurations.

We now present a combinatorial definition of noncrossing. This
definition is expressed as constraints on configurations with an
annotation matrix $(C,A) \in \conf_0(\LL) \times \RR^{E(\LL) \times
E(\LL)}$.  These constraints are equivalent to the constraints placed
on noncrossing configurations
in~\cite{InfinitesimallyLocked_LasVegas}, however, we have translated
combinatorial configurations into our configuration space structure to
help clarify the equivalence.  We will not detail the (fairly
straightforward) correspondence between our combinatorial formulation
and theirs, though we will introduce ``corridor segments'', ``vertex
locations'', and ``magnified views'' that directly relate to the
edges, vertices, and magnified views
in~\cite{InfinitesimallyLocked_LasVegas}.

\begin{definition}

For any configuration $C$ of a linkage $\LL$, we consider a magnified
view around each \defword{vertex location} (i.e., each point at which
at least one vertex is located).  For each vertex location, define the
\defword{inbounds} at that location as follows.  There is one inbound
per nonzero-length bar that has an endpoint co-located with the vertex
location, and two inbounds per bar that goes through the vertex
location. Inbounds to a vertex location are grouped into
\defword{entrances}, the directions from which they are incident.  We
write an inbound as a pair $(\theta,e)$ where $\theta$ is the entrance
and $e$ is the edge.  Two inbounds to a vertex location are
\defword{directly connected} when there is a zero length path in $\LL$
between them, including when they are part of the same bar that passes
through the vertex location.

Figure \ref{fig:combinatorial-def} gives a simple example that will be
helpful for visualizing the situation.  The disks are the magnified
views, the directions from which lines approach the disks from outside
are the entrances, and the intersections between the lines and
boundaries of the disks are the inbounds.  Two inbounds are directly
connected if they are in the same connected component of the graphs
inside the disks.

We define a \defword{combinatorial noncrossing configuration} of a
linkage $\LL$ to be a pair $(C, A) \in \conf_0(\LL) \times \RR^{E(\LL)
\times E(\LL)}$ which satisfies the following constraints:

\begin{enumerate*}

\item \defword{Macroscopically Noncrossing:} Bars $e_i$ and $e_j$
cannot have a strict crossing (they can touch at their endpoints or
overlap over a finite length).

\item \defword{Well-Annotated:} If $i \ne j$ and bars $e_i$, $e_j$
overlap over a nonzero length $l$, then $A_{i,j}=\pm l$.  Otherwise,
set $A_{i,j} = \ord(e_i,e_j)$.\footnote{Note that the annotations in
combinatorial configurations are only meaningful for overlapping bars;
we choose $A_{i,j} = \ord(e_i,e_j)$ for nonoverlapping bars to
simplify the statement of Theorem~\ref{thm:equivalence}.}

\item \defword{Well-Ordered:} At each vertex location $v$, there is a
total ordering $\orel$ on inbounds $(\theta, e_i)$, defined by the
angle of $e_i$ out from $v$ with ties broken by the annotations.  Let
$\dir(e)$ be $+1$ if edge $e$ is directed towards $v$, and $-1$
otherwise. Then $\orel$ is defined as follows:
\begin{equation}
(\theta_i, e_i) \orel (\theta_j, e_j) \iff
\begin{array}{ll}
\theta_i > \theta_j & \textrm{when}\ \theta_i \ne \theta_j \\ A_{i,j}
\dir(e_i) \ge 0 & \textrm{when}\ \theta_i=\theta_j\\
\end{array}
\end{equation}
It follows easily from the definition of $\ord$ that $\dir(e_i)
\textrm{sign}(A_{ij}) = - \dir(e_j)\textrm{sign}(A_{ji})$, so that
this always defines a total ordering.

\item \defword{Microscopically Noncrossing:} The ordering of inbounds
around a vertex location is compatible with the direct connections
between those inbounds.  More precisely, for inbounds $t_1 \orel t_2
\orel t_3 \orel t_4$, if there are direct connections both between
$t_1$ and $t_3$, and between $t_2$ and $t_4$, then all four inbounds 
are directly connected.

\end{enumerate*}
We denote the space of combinatorial noncrossing configurations by
$\cgconf{0}{\LL}$.
\end{definition}



\begin{theorem} \label{thm:equivalence}
The sets of noncrossing configurations, extended noncrossing
configurations, and combinatorial noncrossing configurations are
identical, i.e. $\gconf{0}{\LL} = \igconf{0}{\LL} = \cgconf{0}{\LL}$.
\end{theorem}

There are three inclusions to prove.  $\igconf{0}{\LL} \subset
\gconf{0}{\LL}$ follows directly from the definitions. \ifabstract
\else
$\gconf{0}{\LL} \subset \cgconf{0}{\LL}$ will be proven by showing
that all the conditions in the definition are indeed met, and
$\cgconf{0}{\LL} \subset \igconf{0}{\LL}$ will be shown by
constructing a converging sequence of nontouching configurations.

\begin{definition}
A $\delta$-perturbation of a combinatorial self-touching configuration
$C$ is a nontouching configuration in which each vertex is within
$\delta$ of its location in $C$ and the relative positions of the bars
are preserved.
\end{definition}

\begin{lemma}\label{lemma:equivalence-1}
$\cgconf{0}{\LL} \subset \igconf{0}{\LL}$
\end{lemma}
\begin{proof}
Let $(C,A) \in \cgconf{0}{\LL}$.  By definition, $C \in \conf_0(\LL)$.
Consequently, it suffices to show $(C,A) \in
\overline{\annot{\nconf_{1}(\LL)}{\LL}}$.  By Theorem 3.1 of chapter 1
of Ares Rib\'{o} Mor's Thesis~\cite{Ribo-2006}, for any $C \in
\cgconf{0}{\LL}$, for any $\delta > 0$, there is a nontouching
$\delta$-perturbation $C_\delta$ of $C$.  Because a
$\delta$-perturbation changes bar lengths by at most $2\delta$,
$C_\delta \in \annot{\nconf_{2\delta}(\LL)}{\LL}$.  Because the
relative positions of the bars are preserved in a
$\delta$-perturbation, and the annotation function is continuous for
nontouching configurations, the $C_\delta$ converge to $C$ as $\delta
\rightarrow 0$.  Thus $C \in \overline{\annot{\nconf_{1}(\LL)}{\LL}}
\Rightarrow$ $C \in \igconf{0}{\LL}$, as desired.
\end{proof}

Because the argument of~\cite{Ribo-2006} is quite involved, we provide a
simpler proof that $\cgconf{0}{\LL} \subset \rgconf{0}{\LL}$ to give
some intuition for this result.  The basic strategy is to perturb the
bars within each geometric location containing bars (which we call
corridor segments) so that the bars within the corridor segment are
parallel to each other and are ordered in a consistent fashion.  We
then use the information from the magnified views to implement the
direct connections at the vertex locations.  The details follow.

\begin{definition}
A \defword{corridor} of $\LL$ is a line containing at least one bar of
$\LL$.  A \defword{corridor segment} is an interval in a corridor
which has a vertex location at each end and no other vertex locations
intersecting it.
\end{definition}

In Figure \ref{fig:combinatorial-def-1}, the segments between vertex
locations are the corridor segments, and the two corridor segments
along the bottom combine to form a single corridor.

Because our combinatorial noncrossing configuration is well-annotated,
the annotations define a total order on the bars within each corridor
segment.

\begin{lemma}\label{lemma:total-ordering}
In a combinatorial noncrossing configuration $C$, each corridor $S$
has a total ordering on its bars, that when restricted to any corridor
segment is the order determined by the annotations on that corridor
segment.
\end{lemma}

\begin{proof}
We piece together the ordering for the corridor by proceeding down the
corridor, successively merging the ordering so far with that of each
corridor segment.  At vertex location $v$, we can merge the ordering
so far with the ordering for the next corridor segment if the these
two orderings are consistent.  Because each bar exists for a contiguous
interval along the corridor, it suffices to check that the two
corridor segments of $S$ incident to $v$ have consistent orderings.
The claim that these orderings are consistent is a special case of the
microscopic noncrossing condition at $v$.
\end{proof}

\begin{lemma}\label{lemma:equivalence-1b}
$\cgconf{0}{\LL} \subset \rgconf{0}{\LL}$
\end{lemma}
\begin{proof}
Suppose $(C,A)\in \cgconf{0}{\LL}$.  We will construct a sequence of
nontouching configurations of an extension $\LL'$ of $\LL$, converging
to an extension $(C',A')$ whose reduction to $\LL$ is $(C,A)$.


$\LL'$ is constructed by splitting from each vertex of $\LL$ into one
vertex for each incident edge (with each new vertex incident with its
edge and a zero-length bar to the lexically first new vertex).
Observe that pairs of vertices directly connected in $\LL'$ are
precisely those that are directly connected in $\LL$.  We will call
the zero-length bars extension bars, and the others original bars.
Because both endpoints of each extension bar are endpoints of original
bars, specifying the locations of the original bars defines the
configuration of $\LL'$.  In $C'$, each vertex will (necessarily) lie
in the same place as the vertex of $C$ that it was split from.




Let $0<\delta < \min(1/n,\lmin,(\sin \tmin)/(2n))$ be a real number,
where $\lmin$ is the minimum bar length in $\LL$, $\tmin$ is the
minimum angle between nonparallel bars in $(C,A)$, and $n$ is the
number of bars in $\LL$.


Let $S$ be a corridor with $m$ original bars in it (we treat extension
bars as not belonging to any corridor).  By Lemma
\ref{lemma:total-ordering}, there is a total ordering on the bars in
$S$ compatible with the annotation orderings.  Thus we can assign
distinct offsets $\psi(e) \in \{0, 1, \ldots, m-1\}$ to the bars in
$S$ in a way compatible with the annotation orderings.

Arbitrarily select a unit vector $\overrightarrow{u}$ normal to $S$.
Imagine shifting each original bar $(v,w)$ contained in $S$ from its
location in $C'$ by $\delta^2 \psi((v,w))\overrightarrow{u} $.
Consider also the circle of radius $\delta$ centered at $C'(v)$.
$\delta^2 \psi((v,w)) \le \delta^2n < \delta$, so $v$'s shifted
location is inside this circle.  Because $\delta < \lmin$, $w$'s
shifted location is outside, so the shifted bar intersects $v$'s
circle exactly once.  We set $C_\delta(v)$ to be this unique
intersection of $v$'s circle and shifted bar (and similarly for all
the other vertices in $S$).

We now show that $C_\delta$ is nontouching for sufficiently small
$\delta$.  Original bars never intersect extension bars except at
common vertices because the former lie entirely outside the circles of
radius $\delta$, and the latter entirely inside.  Intersections
between original bars in a common corridor are impossible by
construction.  Because $C_\delta$ converges to $C'$ as $\delta
\rightarrow 0$, original bars that have nonzero separation in $C'$ do
not intersect in $C_\delta$ for small enough $\delta$.  It remains to
handle pairs of original bars that touched in $C'$ precisely at a
vertex location $v$.  Take two such bars, with offsets $i$ and $j$,
and with a relative angle of $\theta$ in $C'$.  If they intersect in
$C_\delta$, it is at a distance at most $(i+j) \delta^2 / \sin
|\theta| \le 2n \delta^2/\sin|\tmin| < \delta$ from $C'(v)$.  Because
neither bar intersects the circle of radius $\delta$ about $v$, no two
original bars cross in $C_\delta$.

We have constructed $C_\delta$ so that the orderings of vertices
around the circles of radius $\delta$ are compatible with the ordering
of inbounds at each vertex location.  The microscopic noncrossing
condition therefore forbids extension bars from crossing.  Thus
$C_\delta$ is noncrossing.

Having shown that $C_\delta$ is noncrossing and converges to $C'$ as
$\delta$ goes to 0, all that remains is to show that $A_\delta$, the
corresponding annotations, converge to $A'$. Because $(C,A)$ is well
annotated, Lemma~\ref{lemma:ord-func-prop} implies each annotation in
$A_\delta$ for pairs of bars not sharing a corridor segment converges
to the corresponding annotation in $A'$.  By
Lemma~\ref{lemma:ord-func-prop}(\ref{lembullet:limit-ord}), the bars
$A_\delta$ for pairs of bars sharing a corridor segment have
accumulation points at $\pm$ the corresponding annotations in $A'$.
But the offsets for bars in the corridors were chosen precisely so the
signs of the annotations in $A_\delta$ matched the signs of
annotations in $A$.  Thus, the annotations converge to $A'$.

Taking any sequence of $\delta$s that converges to zero, we conclude
that $\cgconf{0}{\LL} \subset \rgconf{0}{\LL}$.
\end{proof}

\begin{lemma} \label{lemma:equivalence-2}
$\rgconf{0}{\LL} \subset \cgconf{0}{\LL}$.
\end{lemma}
\begin{proof}
Take any extended noncrossing configuration $(C,A) \in
\rgconf{0}{\LL}$; we need to prove that $(C,A)$ is macroscopically
noncrossing, well-annotated, well-ordered, and microscopically
noncrossing.  Let $(C_k,A_k)$ be a sequence of nontouching
configurations of some extension $\LL'$ of $\LL$ that converges to an
extension $(C',A')$ of $(C,A)$.  The macroscopic noncrossing condition
is easily met because the configurations in which bars have a strict
crossing form an open set, so that a limit of nontouching
configurations cannot have a strict crossing.  The well-annotated
condition follows immediately from
Lemma~\ref{lemma:ord-func-prop}(\ref{lembullet:limit-ord}) and the
continuity of $\ord$ over non-interior-intersecting edges.



To prove the well-ordered and microscopically noncrossing conditions,
we draw small circles around each vertex location.  Take $\delta$
small enough that, in $C'$, the circle of radius $4 \delta$ drawn
around a vertex location does not contain any other vertex locations,
and does not intersect any edges that are not inbounds to the vertex
location.  For some $k_0$ and all $k \ge k_0$, each vertex is less
than $\delta$ away from its final location, so each bar with nonzero
length in $C'$ crosses the circles corresponding to its endpoints, and
each bar with length 0 in $C'$ is contained within the circle that is
common to both its endpoints.  Furthermore, for some $k_1$ and all $k
\ge k_1$, annotations in $A_k$ that have nonzero limits have strictly
the same sign as in $A'$.  Henceforth, we assume that $k \ge k_0,
k_1$.

Suppose $e_i$ is an edge connecting vertex location $v'$ to vertex
location $v$.  Let $R_v$ be the circle centered at $v$ with radius
$2\delta$.  Then $e_i$ intersects $R_v$ somewhere between vertex $v$
and $v'$.  Let $\alpha_{i,k}$ be the angle from a reference direction
to this intersection between $e_i$ and $R_v$.  Without loss of
generality, we may assume the reference direction is not
$\lim_{k\rightarrow\infty} \alpha_{i,k}$ for any bar entering any
vertex location in $\LL$.  Then because for $k$ sufficiently large,
$e_i$ and $e_j$ , there exists $k_2$ such that if
$\lim_{k\rightarrow\infty} \alpha_{i,k} > \lim_{k\rightarrow\infty}
\alpha_{j,k}$, then for all $k \ge k_2$, $\alpha_{i,k} >
\alpha_{j,k}$.  Henceforth, we assume that $k \ge k_2$.

We now define the necessary well-ordering.  We say $Ent(E,e_i) \orel
Ent(E,e_j)$ if for all sufficiently large $k$, $\alpha_{i,k} >
\alpha_{j,k}$.  This is a well-ordering on inbounds at $v$ for $k \ge
k_2$.  Set $\theta_i = \lim_{k\rightarrow\infty} \alpha_{i,k}$.
Inbound edges $e_i$ and $e_j$ share the same entrance $E$ at $v$ if
and only if $\theta_i = \theta_j$.

If $\theta_i > \theta_j$, then because $\lim_{k\rightarrow\infty}
\alpha_{i,k} = \theta_i > \theta_j = \lim_{k\rightarrow\infty}
\alpha_{j,k}$, for all sufficiently large $k$, $\alpha_{i,k} >
\alpha_{j,k}$, and thus the well-ordering condition is satisfied in
this case.

If inbound edges $e_i$ and $e_j$ share a common entrance, then in $c'$
they overlap.  The annotations now give the relationship between
$e_{i}$ and $e_{j}$ in $C_k$.  Assume $e_i$ is oriented from $v$
towards $v'$, and $\ord(e_i, e_j) > 0$ in $C'$ (the other cases are
symmetric).  Then for sufficiently large $k$, in $C_k$, $\ord(e_{i},
e_{j}) > 0$.  Because in $c_k$, $e_{i}$ and $e_{j}$ are nontouching, it
follows from the fact that $\alpha_{i,k} - \alpha_{j,k}$ goes to zero
as $k \rightarrow \infty$ that $\alpha_{i,k} > \alpha_{j,k}$ for
sufficiently large $k$.  This completes the proof of the well-ordering
condition.

Consider now the portion of the linkage in configuration $C_k$ which
is contained inside $R_v$.  This portion of the linkage must be a
planar graph, because $C_k$ is noncrossing.  Two intersection points of
bars with $R_v$ are connected by this graph if and only if the
corresponding inbounds are directly connected.  Given that the order
of the intersection points around $R_v$ matches the order of the
inbounds to $v$, the fact that the graph is planar is precisely the
microscopic noncrossing condition.

Thus $\rgconf{0}{\LL} \subset \cgconf{0}{\LL}$.
\end{proof}

Theorem \ref{thm:equivalence} now follows immediately from Lemmas
\ref{lemma:equivalence-1} and \ref{lemma:equivalence-2} and the fact
that $\LL$ is a (trivial) extension of~$\LL$.

\fi

\section{The Generalized Carpenter's Rule Theorem}
\label{sec:carpenter}

The Carpenter's Rule Theorem says that any nontouching configuration
of an open or closed chain linkage can be convexified through a
continuous motion~\cite{Linkage}.  In this section, we
use our definition of a noncrossing linkage to extend the Carpenter's
Rule Theorem to all noncrossing linkages.  That is, we shall show that
when $\LL$ is an open chain linkage, $\igconf{0}{\LL}$ is connected.

One might hope to show $\igconf{0}{\LL}$ is connected for closed
chains.  Such a result is not true, because the configuration space for
a closed chain may have two connected components, one that turns
clockwise, and one that turns counter-clockwise.  We instead
generalize by showing that any connected component of the
noncrossing configuration space contains a connected component of
the corresponding nontouching configuration space.  To express this
concept, we use the following definition:

\begin{definition}
We say that a semi-algebraic set $A$ \emph{\dominates}a semi-algebraic
subset $B$ if every connected component of $A$ contains a connected
component of $B$.
\end{definition}

We will implicity use in the following discussion that connected
components of semi-algebraic sets are path-connected and
semi-algebraic~\cite[Proposition 2.5.13]{Bochnak-Coste-Roy-1998}.

\begin{lemma}
\label{lemma:sconf-connected}
If $\LL$ is a chain linkage and $\nconf_0(\LL) \neq \emptyset$, and
$\epsilon \ge 0$, then $\nconf_{\epsilon}(\LL)$ \dominates
$\nconf_0(\LL)$
\end{lemma}

\begin{proof}
Each configuration $C$ of a chain linkage has a corresponding
canonical configuration.  For an open chain, it is the straight
configuration; for a closed chain, it is the configuration where the
vertices are concyclic, turning in the same direction as $C$.

For any $C \in \nconf_{\epsilon}(\LL)$, $C$ is connected to its
canonical configuration by the nontouching Carpenter's Rule
Theorem.  Thus, if we can show $C$'s canonical
configuration is connected to an element of $\nconf_0(\LL)$, it will
follow that $\nconf_{\epsilon}(\LL)$ \dominates $\nconf_0(\LL)$.

Fix a vertex location and an edge direction from that vertex (to
factor out translations and rotations).  Then there is a unique map
sending each configuration $C \in \nconf_{\epsilon}(\LL)$ to a
corresponding canonical configuration $C' \in \nconf_0(\LL)$ turning
in the same direction.  By linearly interpolating between the
canonical configuration for $C$ and $C'$, we obtain a path between the
two in $\conf_{\epsilon}(\LL)$.  Because at the endpoints of the path,
each vertex is in convex position, the vertices remain in convex
position along the path, so that the path is contained entirely in
$\nconf_{\epsilon}(\LL)$.  Thus the canonical configuration for $C$ is
in the same connected component as $C' \in \nconf_0(\LL)$, as desired.
\end{proof}

\begin{lemma}
\label{lemma:igconf-connected}
Suppose a linkage $\LL$ is connected, and there exists a $\delta>0$
such that for all $\epsilon \le \delta$, $\nconf_\epsilon(\LL)$
\dominates $\nconf_0(\LL)$.  Then $\igconf{0}{\LL}$ \dominates
$\annot{\nconf_0(\LL)}{\LL}$.
\end{lemma}
\begin{proof}
This lemma is trivially true if for some $\epsilon > 0$,
$\nconf_{\epsilon}(\LL)$ is empty, in which case $\igconf{0}{\LL}$ is
also empty. In the rest of this proof, we assume that this is not the
case.

Suppose $n > 1/\delta$ (so that $0 < 1/n < \delta$).  Because the
annotation function is continuous on nontouching configurations, and
$\nconf_{1/n}(\LL)$ \dominates $\nconf_0(\LL)$, $\asconf{1/n}{\LL}$
\dominates $\annot{\nconf_0(\LL)}{\LL}$.  It follows that the
topological closure $\overline{\asconf{1/n}{\LL}}$ of
$\asconf{1/n}{\LL}$ \dominates $\annot{\nconf_0(\LL)}{\LL}$,
because the closure of a connected set is connected.

Let $Y_1, \ldots, Y_r$ be the connected components of
$\annot{\nconf_0(\LL)}{\LL}$ (semi-algebraic sets always have finitely
many connected components).  Because $\overline{\asconf{1/n}{\LL}}$
\dominates $\annot{\nconf_0(\LL)}{\LL}$, each connected component of
$\overline{\asconf{1/n}{\LL}}$ contains one of the $Y_j$.  Thus, we
can write $$\overline{\asconf{1/n}{\LL}} = \bigcup_{j=1}^{r}
X_{j,n}$$ where $X_{j, n}$ is the connected component of
$\overline{\asconf{1/n}{\LL}}$ containing $Y_j$.  For a given $n$, two
different $X_{j,n}$ are either disjoint or equal.  Further, because
$\overline{\asconf{1/(n+1)}{\LL}} \subset
\overline{\asconf{1/n}{\LL}}$, $X_{j,n+1} \subset X_{j,n}$.  By Lemma
\ref{lem:igconf-property},
\[
\igconf{0}{\LL} = \bigcap_{n=1}^\infty \overline{\asconf{1/n}{\LL}} =
\bigcap_{n=1}^\infty \bigcup_{j=1}^r X_{j,n} =
\bigcup_{j=1}^r \bigcap_{n=1}^\infty X_{j,n}
\]
where we can commute the intersection and union because for each $j$,
the $X_{j,n}$ are a descending sequence.  Now, $\bigcap_{n=1}^\infty
X_{j,n}$ contains $Y_j$, because each $X_{j,n}$ does.  Thus to prove our
lemma, it suffices to show $\bigcap_{n=1}^\infty X_{j,n}$ is
connected.

Notice that the $X_{j,n} \subset
\overline{\annot{\nconf_{1/n}{\LL}}{\LL}}$ are nonempty sets invariant
under translation.  By factoring out the translations via the choice
of a point to place at the origin, we can write $X_{j,n} = \RR^2
\times K_{j, n}$, where $K_{j,n}$ is nonempty, closed, and connected.
Further, because $\LL$ is connected, and our bars have bounded length,
$K_{j,n}$ is bounded in $\RR^N$ for $N = |V| + |E|^2 - 2$.  It follows
that $K_{j,n}$ is compact.  Thus $\bigcap_{n=1}^\infty K_{j,n}$ is the
intersection of a descending sequence of nonempty, compact, connected
sets, and thus is a nonempty compact, connected set.  Thus
\[
\bigcap_{n=1}^\infty X_{j,n} = \bigcap_{n=1}^\infty \RR^2 \times
K_{j,n} = \RR^2 \times \bigcap_{n=1}^\infty K_{j,n}
\]
is the product of connected sets, and hence connected.  Thus
$\igconf{0}{\LL} = \bigcup_{j=1}^r \bigcap_{n=1}^\infty X_{j,n}$ is a
union of connected sets, each containing a connected component of
$\annot{\nconf_0(\LL)}{\LL}$, so $\igconf{0}{\LL}$ \dominates
$\annot{\nconf_0(\LL)}{\LL}$, as desired.
\end{proof}

We are now ready for our main theorem.

\begin{theorem}[Generalized Carpenter's Rule Theorem]

For any open chain linkage $\LL$, $\igconf{0}{\LL}$ is connected.  For
any closed chain linkage $\LL$, $\igconf{0}{\LL}$ has at most two
connected components.

\end{theorem}

\begin{proof}
Suppose $\nconf_0(\LL) \neq \emptyset$.  By the Carpenter's Rule
Theorem, $\nconf_0(\LL)$ has one connected component if $\LL$ is an
open chain, and at most two if $\LL$ is a closed chains.  Applying
Lemmas~\ref{lemma:sconf-connected}, and~\ref{lemma:igconf-connected},
we see that $\igconf{0}{\LL}$ also has this property.

If $\nconf_0(\LL) = \emptyset$, then $\LL$ must be a closed chain.  If
for some $\epsilon > 0$, $\nconf_\epsilon(\LL)$ is empty, then
$\igconf{0}{\LL} \subset \overline{\nconf_\epsilon(\LL)}$ is empty,
and so there are no connected components.  The remaining case is that
$\LL$ has no nontouching configurations, but for any $\epsilon>0$,
there are $\epsilon$-related linkages to $\LL$ with nontouching
configurations.  This happens precisely when $\LL$ is a closed chain
where one edge is equal in length to the sum of all the others, as in
Figure \ref{fig:need-stretch}.  Such $\LL$ have at most two
noncrossing configurations (related by reflection), and thus at most
two connected components.
\end{proof}

\section{Strictly Slender Polygonal Adornments}\label{sec:slender}
\begin{wrapfigure}{r}{1.8in}
  \centering
  \vspace*{-8ex}
  \includegraphics[scale=0.8]{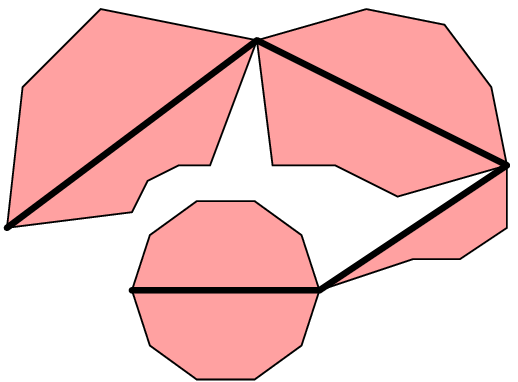}
  \vspace*{-1ex}
  \caption{A polygonally adorned chain with strictly slender adornments.}
  \vspace*{-4ex}
  \label{fig:strictly-slender-polygonal-adornment}
\end{wrapfigure}

In \cite{chains-of-planar-shapes}, it is shown that chains with
\emph{slender adornments} satisfy an analogue of the Carpenter's Rule
Theorem: every such open chain can be straightened and every such
closed chain can be convexified.  In this section, we show that
strictly slender polygonal adornments of open chains satisfy a version
of the self-touching Carpenter's Rule Theorem.

\begin{definition}
A \emph{polygonal adornment} $R$ is a compact, simply connected
polygonal region, together with a \emph{base} $B$, which is a
distinguished line segment connecting two of its boundary points and
contained in $R$.
\iffull

\fi An \emph{inward normal} of a polygonal adornment is a ray $X$
perpendicular to the boundary of the adornment starting from a point
$x \in R \setminus B$ such that $R$ contains a neighborhood of $x$ in
$X$.  \iffull

\fi
A polygonal adornment is \emph{strictly slender} if every inward
normal intersects the relative interior of its base.
\iffull

\fi A \emph{polygonally adorned chain} is set of polygonal adornments
the bases of which form a chain.  The noncrossing configurations of a
polygonally adorned chain are those where the adornments intersect
only on their boundaries.
\end{definition}

See Figure~\ref{fig:strictly-slender-polygonal-adornment} for an
example.

We can model a polygonally adorned linkage $\LL$ as a linkage $\LL'$
by replacing each polygonal adornment with a linkage triangulating
that adornment (we will call the resulting object a ``triangulated
polygonally adorned chain'').  This modeling is faithful because (1)
any triangulation is rigid and (2) if two regions were to move from a
nonoverlapping configuration to an overlapping one, the bars in the
corresponding linkage defined by their boundaries would have to cross.
We can thus re-use the topological machinery of
Lemma~\ref{lemma:igconf-connected} while replacing
Lemma~\ref{lemma:sconf-connected} with an analogue for linkages
obtained by triangulating an open chain adorned by strictly slender
polygonal adornments.

\begin{lemma}\label{lemma:strictly-slender}
Suppose $\LL$ is a triangulated polygonally adorned open chain with
strictly slender adornments.  Then there exists some $\delta > 0$ such
that for all $\epsilon < \delta$, any linkage $\epsilon$-related to
$\LL$ is strictly slender.
\end{lemma}
\begin{proof}
The property of being strictly slender is an open condition on the
edge lengths of $\LL$, and thus there exists a neighborhood of $\LL$
that is contained in the set of strictly slender linkages.
\end{proof}

\begin{lemma}\label{lemma:strictly-slender-connected}
Suppose $\LL$ is a triangulated polygonally adorned open chain with
strictly slender adornments.  Then for some $\delta > 0$,
$\nconf_\epsilon(\LL)$ \dominates $\nconf_0(\LL)$ for all $\epsilon
\le \delta$.
\end{lemma}

\ifabstract
\else

\begin{proof}
Our argument follows the paradigm of Lemma~\ref{lemma:sconf-connected}.  The canonical configurations are those
in which the chain is straight (if the chain has $n$ adornments, there
are potentially $2^n$ such canonical configurations, determined by the
choices of reflection for each adornment).  Thus, each configuration
of $\LL$ has a corresponding canonical configuration.

Suppose $\epsilon < \delta$, for the $\delta$ defined in
Lemma~\ref{lemma:strictly-slender}.  Then for any $C \in
\nconf_{\epsilon}(\LL)$, $C$ is connected to its canonical
configuration by Theorem 8 of~\cite{chains-of-planar-shapes}.  Thus,
if we can show $C$'s canonical configuration is connected to an
element of $\nconf_0{\LL}$, it will follow that
$\nconf_{\epsilon}(\LL)$ \dominates $\nconf_0(\LL)$.

Because we have an open chain, there is a path linearly interpolating
between the canonical configuration for $C$ and the corresponding
canonical configuration in $\nconf_0(\LL)$.  This path is contained
entirely in $\nconf_{\epsilon}(\LL)$, because two different slender
adornments with bases in straight configuration never touch except at
the endpoints of the bases, and within each triangulated adornment,
for sufficiently small $\epsilon$ there will be no crossings.  Thus
the canonical configuration for $C$ is in the same connected component
as $C' \in \nconf_0(\LL)$, as desired.
\end{proof}
\fi

\begin{theorem}\label{thm:slender-adornments}
Suppose $\LL$ is triangulated polygonally adorned open chain with
strictly slender polygonal adornments.  Then any configuration in
$\igconf{0}{\LL}$ can straighten its base.
\end{theorem}

\begin{proof}
The result follows from Lemmas~\ref{lemma:igconf-connected}
and~\ref{lemma:strictly-slender-connected}, because any connected
component of $\nconf_0(\LL)$ contains a configuration in which its
base is straight by~\cite{chains-of-planar-shapes}.
\end{proof}

Theorem~\ref{thm:slender-adornments} may be somewhat unsatisfying in
that it has a number of restrictions on its applicability: ``open
chain'', ``polygonal'', and ``strictly slender'', in particular.
Theorem~\ref{thm:slender-adornments} is not true for closed chains,
because while every slenderly adorned closed chain can be convexified,
the convex configurations of closed adorned chains are not necessarily
reachable from each other~\cite{journal-chains-of-planar-shapes}.  The
polygonal restriction is not important at all; one can extend
Theorem~\ref{thm:slender-adornments} to arbitrary strictly slender
adornments via a limit of polygonal adornment approximations.

It remains open whether Theorem~\ref{thm:slender-adornments} holds in
the case of non-strictly slender adornments.  The strictly slender
restriction is fundamental to our argument, because
Lemma~\ref{lemma:strictly-slender} is false if one replaces ``strictly
slender'' with ``slender''.  To see this, suppose $\LL$ is a
triangulated chain with a right-triangle adornment.  There are
linkages $\epsilon$-related to $\LL$ that are not slender (one needs
only shift the vertex at the right angle slightly to give it a
non-inward normal).

\section{Extensions of the Definition Methodology}
\label{sec:extension}

In this section, we outline how our methodology for defining self-touching
configurations could be extended to other types of objects.
The definitions suggested in this section are preliminary.

\subsection{Polygonal Assemblies and Rigid Origami}

A \defword{polygonal assembly} is a set of polygons in 3-space and a
relation indicating which polygonal edges are attached together.
Polygonal assemblies arise in the study of \defword{rigid origami},
where the edges correspond to creases.

Polygonal assemblies are a natural generalization of 2D linkages to
3D.  Edges are replaced by polygons, and vertices are replaced by
edges.  We can define $\ord(P_1, P_2)$ as the area of the projection
of $P_2$ onto $P_1$, signed by which side of $P_1$ $P_2$ is on, in
direct analogy with the linkage definition.  With this definition of
$\ord$, we can extend universal reconfigurability results for
single-vertex rigid origami~\cite{Streinu-Whiteley-2004} to self-touching
configurations (i.e., those in which two sheets are folded flat
against each other).  As in the case of Slender Adornments, we can
only prove universal reconfigurability for self-touching
configurations that are in an open subset of all configurations, in
this case those for convex cones.

\subsection{3D Linkages}

It does not seem possible to model 3D linkages with our methodology.
The difficulty is that the codimension of object elements is 2 for 3D
linkages (compared with 1 for 2D linkages and polygonal assemblies).
Consequently, there is a continuum of ways in which two bars can
overlap (each relative direction is possible), and thus it seems no
function has the necessary continuity properties to define the
annotations.

\subsection{Paper}

Paper is a much more interesting challenge for our definition
methodology than linkages or polygonal assemblies.  Indeed, paper has
an infinite-dimensional configuration space, so we have to worry about
the right topology to use.  Moreover, with paper, the individually
movable pieces are infinitesimally small, so an order function that is
zero when it is not directly above or below a piece would be zero
everywhere.
\ifabstract
\else

We work with a unit-size $n$-dimensional closed sheet of paper in
$(n+1)$-dimensional space.  A (possibly self-crossing) configuration
of order $k$ of a sheet of paper is represented by a mapping $f$ from
$[0,1]^n$ to $\RR^{n+1}$. The order $k$ of the configuration indicates
the regularity of the mapping; $f$ must be piecewise $C^k$ except
along a finite set of $C^k$ hyper-surfaces of finite hyper-area. To
avoid stretching the paper, $f$ must also be an isomorphism, i.e.,
wherever it is defined, its Jacobian must be an orthogonal projection
of rank $n$.

A nontouching configuration is simply a configuration for which $f$ is
injective. We now consider an example order function; this is preliminary
work.

\subsubsection{``Distance with Obstacles'' Order Function}

First we consider the following order function that maps two points on the
paper to a real number:
\iffull
$$ \ord(a, b) = d_o(a_+, b) - d_o(a_-, b), $$
\else
$ \ord(a, b) = d_o(a_+, b) - d_o(a_-, b) $,
\fi
where
$d_o(a_+,b)$ is the infimum of the lengths of the paths that start from the
positive side of the paper at $a$ and end at $b$ without crossing the
paper. This function is nice because it is continuous when the
configuration is varied in a nontouching way, and when the two points $a$ and
$b$ converge towards each other in a sequence of nontouching
configurations, the order function converges to a limit that depends on
the side of the paper from which $b$ converges to $a$. The annotation
function is produced by applying the order function to each pair of points
on the paper. This defines a set of annotated nontouching configurations.

Before we can define noncrossing configurations, we need to specify a
distance function that will define the topology we are using when
taking limits. We define this topology over all functions like $f$,
except that we do not impose the isomorphism constraint. The distance
between $f$ and $g$ is defined by: $$ d(f,g) = \max(\sup
|f(x)-g(x)|,\quad \sup \|D_f(x)-D_g(x)\|, \quad \sup
|\ord(f(x),f(y))-\ord(g(x),g(y))|) $$ where $x$ and $y$ are in
$[0,1]^n$ and $D_f(x)$ is the Jacobian of $f$ around $x$ divided
by the norm of the second order derivative of $f$ around
$x$. Essentially, $d$ is the supremum norm applied to the functions,
their derivatives and their annotations, except that the derivatives
are scaled by a factor inversely proportional to the second order
derivative of each function around that point (multiplied by zero at
crease points).  Without this scaling, points that should converge
near creases would cause convergence problems.\footnote{This problem
could be solved in other ways such as by using an integral norm
instead of the supremum norm (might allow kinks to remain in the
paper), or by considering that the crease must form at its final
location after a finite number of steps in the limit taking process.}

We can now define a noncrossing configuration as a limit of annotated
noncrossing configurations for which the isometry constraint has been
dropped. We drop the isometry constraint to ensure that we do not miss any
self-touching configurations.

The order function we have chosen is somewhat tedious to work with because
of its global nature. It would be nice to find a definition of the order
function that only depends on $f$ in a neighborhood of the points to which
it is being applied. This is a difficult task because the paper can have
very small features, so there is no canonical neighborhood size to take.
One possibility might be to pick the largest neighborhood over which $f$
behaves ``nicely''. This is all left for future work.
\fi

\subsection{Conclusion}

In this paper, we have introduced a new topological definition of
self-touching 2D linkages.  It is equivalent to the previously proposed
definition, but is easier to work with.  In particular, we have shown how to
use it to extend the Carpenter's Rule Theorem to self-touching linkages.

This result demonstrates the advantages of our topological definition
over the previously proposed combinatorial definition.  The underlying
topological methodology can be used to define the self-touching
configurations of other classes of codimension-$1$ self-touching objects
such as origami and polygonal assemblies in terms of their nontouching
configurations.  Objects of codimension $2$ or more seem more
difficult to characterize.

\let\realbibitem=\bibitem
\def\bibitem{\par \vspace{-1.2ex}\realbibitem}

\begin{small}
\bibliographystyle{alpha}
\bibliography{paper}
\end{small}

\end{document}